\documentclass[aps, prl, reprint]{revtex4-2}
\usepackage{amsfonts}
\usepackage{amsmath}
\usepackage{amsthm}
\usepackage{braket}
\usepackage{tikz}
\usetikzlibrary{graphs}
\usepackage{subcaption}
\usepackage{caption}
\newtheorem{theorem}{Theorem}
\newtheorem{proposition}[theorem]{Proposition}
\newtheorem{lemma}[theorem]{Lemma}
\newtheorem{corollary}[theorem]{Corollary}

\newcommand{\mbb}{\mathbb}
\newcommand{\Tr}{\text{Tr}}
\newcommand{\id}{\mbb{I}}

\newcommand{\Br}{\mathfrak{B}_t(d)}
\newcommand{\Or}{\mbb{O}(d)}
\newcommand{\Un}{\mbb{U}(d)}
\newcommand{\pr}{\text{pr}}
\newcommand{\sh}{S^d(\mbb{C})}
\newcommand{\sy}{\mbb{S}_t}
\newcommand{\mt}{\mathcal{M}_t}
\newcommand{\mf}{\mathfrak}

\begin{document}
\title{Random Real Valued and Complex Valued States Cannot be Efficiently Distinguished}

\author{Louis Schatzki}
\email{louisms2@illinois.edu}
\affiliation{Department of Electrical and Computer Engineering, Coordinated Science Laboratory, University of Illinois at Urbana-Champaign, Urbana, Illinois 61801, USA}

\begin{abstract}
    In this short note we show that the ensemble $\{O\ket{0}\bra{0}O^\top \ \vert \ O \in \Or\}$, where $O$ is drawn from the Haar measure on $\Or$ cannot be distinguished from $t$ copies of a Haar random state unless $t = \Omega(\sqrt{d})$. Our proof has the benefit of \textit{exactly} computing the trace distance, which scales as $\Theta(t^2/d)$ for $t = O(\sqrt{d})$, between the moments as well as being surprisingly short. Lastly, we show that twirling certain states with orthogonal matrices yields exact $t=3$ designs, yet the same cannot be true for $t>3$.
\end{abstract}

\maketitle

\section{Introduction}
Randomness plays an important role in quantum information theory. Of particular interest are random quantum states, which are drawn uniformly from the complex unit sphere. Such objects show up in  randomized benchmarking and measurements \cite{dankert2009exact, elben2023randomized, nakata2021quantum, PhysRevX.9.021061, choi2023preparing}, quantum chaos \cite{Roberts_2017, PRXQuantum.4.010311}, black holes \cite{sekino2008fast}, and analyzing variational quantum machine learning \cite{mcclean2018barren, holmes2022connecting},.

While truly random states are not easy to produce, for many protocols using random states, it suffices to sample finite ensembles that replicate random states for low-degree polynomials. Such ensembles are known as state $t$-designs and have been heavily studied in quantum information \cite{gross2007evenly, ambainis2007quantum, PRXQuantum.4.010311}.

In recent years, much work has gone into understanding pseudorandom states \cite{ji2018pseudorandom, ananth2022cryptography, ananth2022pseudorandom}. These are states that, while being efficiently generable, cannot be efficiently distinguished from truly random quantum states. Several constructions are known, such as binary phase states \cite{brakerski2019pseudo} and subset states \cite{fern2023pseudorandom, giurgicatiron2023pseudorandomness}. Proving the computational indistinguishability of these constructions starts by showing statistically indistinguishability of truly random phase/subset states. This implies that the ``realness" of a random quantum state cannot be efficiently tested \cite{haug2023pseudorandom}. By contrast, with black box access to a unitary one can efficiently test for ``realness" \cite{haug2023pseudorandom}. Though it would be surprising, this does not preclude the possibility that Haar random real-valued states may be efficiently distinguished from complex-valued states.

It was recently shown \cite{west2024random} that the symplectic group forms an exact state $t$-design for any $t$. Inspired by that result, in this note we ask what happens when we draw states at random from an orbit of the Orthogonal group. Using the representation theory of the Orthogonal group, finding the exact trace distance between random real-valued and complex-valued states turns out to be surprisingly straight-forward.

We prove the following two results.

\begin{theorem}\label{thm:main}
    The trace distance between the moment operators for $t$ copies of a Haar random real valued state vs complex valued state is $ 1-\prod_{j=1}^{t-1} \frac{d+j}{d+2j}$. If $t < \sqrt{d}$, then the trace distance scales as $\Theta(t^2/d)$.
\end{theorem}

This implies that any protocol to distinguish random real-valued states from complex values states must use $t=\Omega(\sqrt{d})$ samples. We also show the following result for exact state $t$-designs.
\begin{theorem}\label{thm:main2}
    There exist states $\ket{\psi}\in\mbb{C}^d$ such that $\int_{\Or} (O\ket{\psi}\bra{\psi}O^\top)^{\otimes 3} d\mu_O(O)$ forms an exact state $3$-design. The same is not possible for any $t>3$.
\end{theorem}

\section{Preliminaries}
Let $\Un$ be the unitary and $\Or$ the orthogonal group of dimension $d$, the latter being identified with all real-valued matrices in $\Un$. As compact groups, both admit Haar measures $\mu_{U}$ and $\mu_O$. A state $t$-design, $d\eta(\ket{\psi})$, is a distribution on $\mbb{C}^d$ such that the $t^{\text{th}}$ moments match that of the Haar measure on $\sh = \{\ket{\psi} \in \mbb{C}^d\}$ inherited from $\mu_U$. That is,
\begin{align}
    \int_{\sh}\ket{\psi}\bra{\psi}^{\otimes t} d\eta(\ket{\psi}) = \int_{\Un} (U\ket{0}\bra{0}U^\dagger)^{\otimes t}d\mu_U(U)\ .
\end{align}
Note that the choice of $\ket{0}$ on the right hand side is without loss of generality from the invariance of the Haar measure and since $\Un$ acts transitively on $\sh$. A very useful fact is that the right hand side is exactly the (normalized) projector onto the symmetric subspace \cite{harrow2013church}:
\begin{align}
    \int_{\Un} (U\ket{0}\bra{0}U^\dagger)^{\otimes t}d\mu_U(U) & = \frac{1}{P(d,t)}\sum_{\sigma \in \sy}\sigma\ ,
\end{align}
where here by $\sigma$ we implicitly mean the standard representation on $\left( \mbb{C}^d \right)^{\otimes t}$ given by
\begin{align}
    \sigma \bigotimes_{i=1}^t \ket{\psi_i} = \bigotimes_{i=1}^t \ket{\psi_{\sigma^{-1}(i)}}\ ,
\end{align}
and $P(d,t) = \prod_{j=0}^{t-1}d+j$. For notational simplicity, we will denote this state as $\rho_{sym}$ going forward. A distribution $d\eta(\ket{\psi})$ is an additive $\varepsilon$-approximate $t$-design if 
\begin{align}
    \left\Vert \int_{\sh}\ket{\psi}\bra{\psi}^{\otimes t} d\eta(\ket{\psi}) - \rho_{sym}\right\Vert_{1} \leq \varepsilon\ .
\end{align}
Equivalently, no measurement can distinguish between $t$ copies of a Haar random state and $t$ copies of a draw from an approximate state design with bias greater than $\varepsilon/2$.

We will be interested in continuous distributions resulting from Haar random orthogonal matrices. That is, we consider the ensemble $\{ O\ket{\psi}\bra{\psi}O^\top d\mu_O(O)\vert O \in \Or \}$. Since $\Or$ does not act transitively on $\sh$, we cannot take $\ket{\psi}$ to be $\ket{0}$ without loss of generality. In particular, the orbit of $\ket{0}$ is all real-valued states. Thus, we would not expect this ensemble to match that stemming from the unitary Haar measure. However, we will show that the trace distance of the twirled state from that of the unitary group is $\Theta(t^2/d)$, implying that these ensembles cannot be easily distinguished for $t=o(\sqrt{d})$.

To do so, we will require the technique of Weingarten calculus, which we briefly review and refer the reader to \cite{collins2009some, Collins2022, mele2024introduction} for more details. A fundamental observation in representation theory is that $\int_G \varphi(g) X \varphi(g)^{-1} dg$ (where $dg$ is the Haar measure on $G$) projects $X$ onto the commutant of the representation $\varphi$. Letting $G$ be a matrix group and $\varphi(g)=g^{\otimes t}$, we are particularly interested in the $t^{\text{th}}$ order commutant
\begin{align}
    C^t(G) = \{ X \in \mathcal{B}( \left( \mbb{C}^d \right)^{\otimes t}) \ \vert \ [X, g^{\otimes t}] = 0 \ \forall g \in G\}\ .
\end{align}
Let $\{Y_i\}_i$ be some basis for $C^t(G)$. Then,
\begin{align}
    \int_G g^{\otimes t} X \left( g^{-1} \right)^{\otimes t} dg & = \sum_{i} c_i Y_i\ .
\end{align}
We will sometimes label these coefficients $\{c_i\}_i$ with the corresponding operators, i.e. $c_i \equiv c_{P_i}$, and denote the entire vector by $\vec{c}$. Let $G = (\Tr[Y_i^\dagger Y_j])_{i,j}$ be the Gram matrix for the chosen basis. Then $G\vec{c} = \vec{b}$, where $b_i = \Tr[Y_i^\dagger X]$. The Weingarten matrix $W$ is defined as the (pseudo) inverse of $G$. Starting from $\vec{b}$, we can then compute $\vec{c}$ via $W\vec{b}$.

The famous Schur-Weyl duality states that $C^t(\Un) \cong \mbb{C}[\mbb{S}_t]$, where here $\mbb{S}_t$ is identified with its standard representation on $\left( \mbb{C}^d \right)^{\otimes t}$. A similar duality holds with $\Or$ and the Brauer algebra \cite{collins2009some}. We will briefly describe this algebra before proceeding to our results.

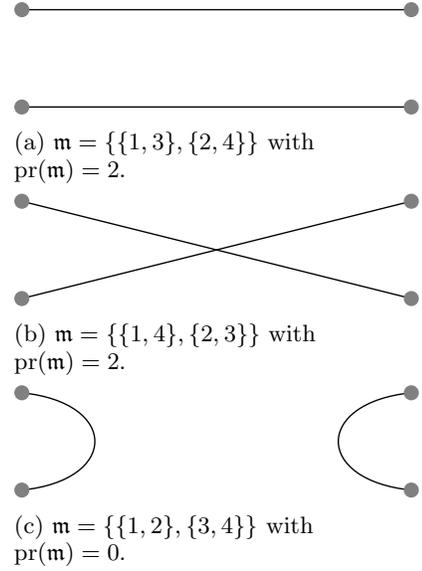
\begin{figure}[t]

    \begin{subfigure}[t]{0.30\textwidth}
          \centering
          \resizebox{\linewidth}{!}{\begin{tikzpicture}
         \draw (-2,0) -- (2,0);
         \draw (-2,1) -- (2,1);
          \filldraw [gray] (-2,0) circle (2pt);
         \filldraw [gray] (2,0) circle (2pt);
         \filldraw [gray] (-2,1) circle (2pt);
          \filldraw [gray] (2,1) circle (2pt);
          \end{tikzpicture}}  
          \caption{$\mf{m} = \{ \{1,3\}, \{2,4\} \}$ with $\pr(\mf{m}) = 2$.}
          \label{fig:A}
     \end{subfigure}
     
     \begin{subfigure}[t]{0.30\textwidth}
          \centering
          \resizebox{\linewidth}{!}{\begin{tikzpicture}
            \draw (-2,0) -- (2,1);
            \draw (-2,1) -- (2,0);
             \filldraw [gray] (-2,0) circle (2pt);
            \filldraw [gray] (2,1) circle (2pt);
            \filldraw [gray] (2,0) circle (2pt);
            \filldraw [gray] (-2,1) circle (2pt);
            \end{tikzpicture}}  
          \caption{$\mf{m} = \{ \{1,4\}, \{2,3\} \}$ with $\pr(\mf{m}) = 2$.}
          \label{fig:B}
     \end{subfigure}
     
     \begin{subfigure}[t]{0.30\textwidth}
          \centering
          \resizebox{\linewidth}{!}{\begin{tikzpicture}
            \draw (-2,0) .. controls (-1,0.1) and (-1,0.9) .. (-2,1);
            \draw (2,0) .. controls (1,0.1) and (1,0.9) .. (2,1);
            \filldraw [gray] (2,1) circle (2pt);
            \filldraw [gray] (-2,0) circle (2pt);
            \filldraw [gray] (-2,1) circle (2pt);
             \filldraw [gray] (2,0) circle (2pt);
        \end{tikzpicture}}  
          \caption{$\mf{m} = \{ \{1,2\}, \{3,4\} \}$ with $\pr(\mf{m}) = 0$.}
          \label{fig:C}
     \end{subfigure}
    
         \caption{The elements of $\mathcal{M}_2$, the set of pair partitions of 4 objects. Their representation on $\mbb{C}^d \otimes \mbb{C}^d$ can be obtained by viewing the above graphs as tensor networks diagrams.}

     \label{fig:1}
\end{figure}

By $\mt$ we denote the set of all pairings of $[2t]$. That is, all partitions of $2t$ objects into sets of size $2$. These have a natural graphical interpretation as bipartite graphs of degree $1$. The case of $\mathcal{M}_2$ is illustrated in fig~\ref{fig:1}. It is standard to write such a partition as
\begin{align}
    \mf{m} = \{ \{ \mf{m}(1), \mf{m}(2) \}, \{ \mf{m}(1), \mf{m}(2) \}, \cdots, \{ \mf{m}(2t-1), \mf{m}(2t) \} \}\ ,
\end{align}
where $\mf{m}(2j-1) < \mf{m}(2j)$ and $\mf{m}(1) < \mf{m}(3) < \cdots < \mf{m}(2t-1)$. The propagating number, $\pr(\mf{m})$, is given by the number of pairs $(2j-1, 2j)$ such that $\mf{m}(2j-1) \leq t$ and $\mf{m}(2j) > t$. Graphically, this is the number of edges across the bipartition.

Let us define a multiplication of $\mf{m}, \mf{n} \in \mt$ by composing their diagrams and adding a factor of some scalar $\delta$ for any loops formed in this process. Then, the Brauer algebra is the associative algebra $\mf{B}_t(\delta) = \mbb{Z}[\delta][\mt]$, where $\mbb{Z}[\delta]$ denotes polynomials in $\delta$ with integer cofficients, under this multiplication operation. Now, $C^t(\Or)$ is a representation of $\Br$ where
\begin{align}
    \mathfrak{m} \mapsto \sum_{i_1, \ldots, i_{2t}=1}^d \ket{i_1\cdots i_{t}}\bra{i_{t+1}\cdots i_{2t}} \prod_{j=1}^t \delta_{i_{\mf{m}(2j-1)}, i_{\mf{m}(2j)}}\ .
\end{align}
It is worth noting that this operator is the same as taking the graph of $\mf{m}$ to represent a tensor network diagram. Let $\gamma = \sum_{i,j} \ket{ii}\bra{jj}$ be the unnormalized maximally entangled state on $\mbb{C}^d \otimes \mbb{C}^d$. Then, we can equivalently write the representation of $\mathfrak{m}$ as $\pi (\gamma^{\otimes t - \pr(\mf{m})} \otimes \id^{\otimes \pr(\mf{m})})\tau$, where $\pi, \tau \in \sy$. With this identification, it is not hard to see that
\begin{align}
    \Tr[\mf{m} \ket{\psi}\bra{\psi}^{\otimes t}] = \vert \langle \psi^* \vert \psi \rangle \vert^{t-\pr(\mf{m})}\ .
\end{align}

It is also useful to realize $\mt$ as a subset of $\mbb{S}_{2t}$ through the mapping
\begin{align}
        \mf{m} \in \mt \mapsto \begin{pmatrix}
        1 & 2 & \cdots & 2t\\
        \mf{m}(1) & \mf{m}(2) & \cdots & \mf{m}(2t)
\end{pmatrix}\ .
\end{align}

The hyperoctahedral group, $H_t \cong \mbb{S}_2 \wr \mbb{S}_t$, is defined as the centralizer of $(1\ 2)\cdots (2t-1\ 2t)$ in $\mathbb{S}_{2t}$ and plays an important role in the theory of Brauer algebras \cite{collins2009some}. In particular, $\mt$, as a subset of $\mbb{S}_{2t}$, forms a transversal set for the left cosets of $H_t$ in $\mbb{S}_{2t}$.

\section{Approximate Designs from Random Real-Valued States}
In this section we consider the ensemble $\{ O\ket{0}\bra{0}O^\top d\mu_O(O) \ \vert \ O\in\Or \}$. First, a technical lemma we will require.
\begin{lemma}\label{lem:trace}
    The sum of the characters of the representation of the basis of the Brauer algebra is given by
    \begin{align}
        \sum_{\mathfrak{m} \in \Br(d)} \Tr[\mathfrak{m}] & = \prod_{j=0}^{t-1}d+2j =: Z(d,t)\ .
    \end{align}
\end{lemma}
\begin{proof}
Let $e = \frac{1}{2^t t!}\sum_{\zeta \in H_t}\zeta \in \mbb{C}[\mbb{S}_{2t}]$. In \cite{collins2009some} the authors derive the Fourier expansion of the gram matrix $\hat{G}$, where $\hat{G}$ is viewed as an element of $e\mbb{C}[\mbb{S}_{2t}]e$ (which is naturally ismorphic to the space of $H_t$ bi-invariant functions on $\mbb{S}_{2t}$). Therein they show that
\begin{align}
    \hat{G} = \frac{1}{(2t)!}\sum_{\substack{\lambda \vdash t\\ l(\lambda) \leq d}}f^{2\lambda}Z_\lambda(1^d)\omega^\lambda\ ,
\end{align}
where $f^{2\lambda}$ is the hook length formula, $Z_\lambda$ is the zonal function, and $\omega^\lambda=\chi^{2\lambda} e$ is the zonal spherical function (where $\chi^{2\lambda}$ is an irreducible character of $\mbb{S}_{2t}$). Further,
\begin{align}
    Z_\lambda(1^d) & = \prod_{(i,j)\in \lambda}(d+2j-i-1)\ ,
\end{align}
which, for $\lambda = (t)$, is exactly $Z(d,t)$ as in the lemma statement. Using the fact (lemma 4.3 of \cite{collins2009some}) that $\omega^\lambda \omega^\mu = \delta_{\lambda,\mu}\frac{(2t)!}{f^{2\lambda}}\omega^\lambda$, we find that
\begin{align}
    \omega^{(t)}\hat{G} & = f^{(2t)}Z(d,t)\omega^{(t)} = Z(d,t)e,
\end{align}
which is the convolution of the Gram matrix with the identity operator. The scalar factor is then exactly the sum we are after.
\end{proof}

\begin{proposition}\label{prop:real-moment}
    Starting with any real-valued state, $\ket{0}$ without loss of generality, the moment operators under conjugation with Haar random orthogonal matrices are
    \begin{align}
        \int_{\Or} (O\ket{0}\bra{0}O^\top)^{\otimes t} = \frac{1}{Z(d,t)}\sum_{\mf{m} \in \Br} \mf{m} =: \rho_{br}\ .
    \end{align}
\end{proposition}
\begin{proof}
    Since $\ket{0}$ is real-valued, $\langle 0^* \vert 0 \rangle = 1$ and $\Tr[\mf{m} \ket{0}\bra{0}^{\otimes t}] = 1$ for all $\mf{m} \in \Br$. Thus, the coefficient vector is $b = \vec{1}$. It remains to show that $\vec{1}$ is an eigenvector of the Weingarten matrix. Instead, we show that $\vec{1}$ is an eigenvector of the gram matrix $G$. Since the Weingarten matrix is the (pseudo) inverse of $G$, this would imply that $\vec{1}$ is an eigenvector of the Weingarten matrix as well. If all rows and columns sum to the same value, then $\vec{1}$ is an eigenvalue of $G$. This follows from the Gram matrix being $H_t$ bi-invariant, which we will spell out explicitly here. To prove this, we use the mapping of $\mt$ into $\mbb{S}_{2t}$. Now, it turns out that $(\mathbb{S}_{2t}, H_t)$ forms a Gelfand pair and
\begin{align}
    \mbb{S}_{2t} & = \bigsqcup_{\mf{m} \in \Br} \mf{m} H_t = \bigsqcup_{\rho \vdash t} H_\rho\ ,
\end{align}
where $H_\rho$ is a double-coset of $H_t$ labeled by a partition $\rho$. Further, the Gram matrix of $\Br$ can be evaluated by $G(\mf{m}, \mf{n}) = d^{l(\Xi(\mf{m}^{-1} \mf{n}))}$, where $\Xi(\sigma)$ is the partition type in the double-coset decomposition and $l(\cdot)$ is its length \cite{collins2009some}. Consider the map $\pi \mapsto \mf{m}\pi$ in $\mbb{S}_{2t}$. Since $\mt$ is a transversal set for $H_t$, it follows that $\mf{m}\mf{n} = \mf{p} h$ for $\mf{p} \in \mathfrak{B}_t$ and $h\in H_t$ with the property that $\mf{m} \mf{n} H_t = \mf{m} \mf{p} H_t$ iff $\mf{n} = \mf{p}$. Thus, we can associate with $\sigma$ a set permutation which "shuffles" the elements of $\mt$. In addition, it holds that $\Xi(\mf{m} \mf{n}) = \Xi(\mf{p})$ and $\Xi(\mf{m}^{-1}\mf{m} \mf{p}) = \Xi(\mf{p}) = \Xi(\mf{m}^{-1}\mf{p})$. Since our reshuffling was 1-to-1, this proves that $\sum_{\mf{n} \in \mt} G(\mf{m}, \mf{n})$ is some fixed value for any $\mf{m} \in \mt$. From here it follows that $\vec{1}$ is a eigenvector of $G$ and $G^{-1}$. By lemma~\ref{lem:trace} the normalization must be $1/Z(d,t)$.
\end{proof}

Now that we exactly know the moment operators for the real-valued orbit of $\Or$ and $\Un$, we can compute the trace distance. It turns out that this is relatively straight-forward due to the following lemma.

\begin{lemma}\label{lem:positive}
    The operator $\sum_{\mf{m} \in \Br \backslash \mbb{S}_t} \mf{m}$ is positive semidefinite.
\end{lemma}
\begin{proof}
    Any $\mf{m} \in \Br \backslash \mbb{S}_t$ can be written as $\mf{m} = \pi (\gamma^{\otimes t-\pr(\mf{m})}\otimes \id^{\otimes \pr(\mf{m})})\tau$, where $\pi, \tau \in \mbb{S}_t$. Thus, it follows that
    \begin{align}
        \sum_{\substack{\mf{m} \in \Br \backslash \mbb{S}_t \\ \pr(\mf{m}) = w}} \mf{m} & \propto \left(\sum_{\pi \in \mbb{S}_t}\pi\right)(\gamma^{\otimes t-w}\otimes \id^{\otimes w})\left(\sum_{\pi \in \mbb{S}_t}\tau\right)\ .
    \end{align}
    Since $\gamma$ is positive semi-definite, this proves the lemma.
\end{proof}

Putting the pieces together, we can now exactly derive the trace distance between the moment operators.

\begin{proposition}\label{prop:td}
    The trace distance between the moment operators for drawing Haar randomly from the Unitary group versus the real-orbit of the Orthogonal group is exactly
    \begin{align}
       \Vert \rho_{br} - \rho_{sym} \Vert_{tr} & = 1 - \prod_{j=1}^{t-1}\frac{d+j}{d+2j}\ .
    \end{align}
\end{proposition}
\begin{proof}
    Of course, $\sy \subset \Br$ and $Z(d,t) > P(d,t)$ (for $t > 1$). It follows that
    \begin{align}
        \rho_{br} - \rho_{sym}  = \frac{1}{Z(d,t)}\sum_{\mf{m} \in \Br}\mf{m}-\frac{1}{P(d,t)}\sum_{\sigma\in\mbb{S}_t} \sigma\\
        = \frac{1}{Z(d,t)}\sum_{\mf{m} \in \Br\backslash \sy}\mf{m} - \left( \frac{1}{P(d,t)} - \frac{1}{Z(d,t)} \right)\sum_{\sigma\in\mbb{S}_t} \sigma\ .
    \end{align}
    By lemma~\ref{lem:positive}, we have split $\rho_{br} - \rho_{sym}$ into positive semidefinite and negative semidefinite parts. Then, the trace distance is simply the trace of one of these operators (it does not matter which one since $\Tr[\rho_{br}-\rho_{sym}] = 0$).
\end{proof}

We now give asymptotically matching upper and lower bounds on the trace distance.

\begin{corollary}\label{cor:approx_t_design}
  Let $t  < d/2$. Then, the trace distance between $\rho_{or}$ and $\rho_{sym}$ is $1-\exp\{-\Theta(t^2/d)\}$\ .
\end{corollary}
\begin{proof}
 First, a lower bound:
  \begin{align}
      1-\prod_{j=1}^{t-1}\frac{d+j}{d+2j} & > 1-\prod_{j=1}^{t-1} \left(1-\frac{j}{3d} \right)\\
      & > 1-\exp\left\{-\sum_{j=1}^{t-1}\frac{j}{3d}\right\}\\
      & = 1-\exp\left\{-\frac{t(t-1)}{6d} \right\}\ ,\\
  \end{align}
  where in the first line we have used that $t < d$. Next, an upper bound:
  \begin{align}
       1-\prod_{j=1}^{t-1}\frac{d+j}{d+2j} & < 1-\prod_{j=1}^{t-1}\left( 1-\frac{j}{d} \right)\ .
  \end{align}
  By assumption, $j/d \in [0,0.5]$. We can then apply the inequality $e^{-x} \leq 1-x/2$ for $x\in [0,1.5]$ to obtain
  \begin{align}
      1-\prod_{j=1}^{t-1}\frac{d+j}{d+2j} & < 1-\exp\left\{ -2\sum_{j=1}^{t-1}\frac{j}{d}  \right\}\\
      & = 1-\exp\left\{ -\frac{t(t-1)}{d} \right\}\ .
  \end{align}
\end{proof}
Say that actually $t < \sqrt{d}$. Then, we can again apply the inequalities $1+x\leq \exp{x}$ and $\exp{-x} \leq 1-x/2$ to obtain upper and lower bounds of $t(t-1)/d$ and $t(t-1)/12d$, proving that the trace distance scales as $\Theta(t^2/d)$ in this regime. It is then immediate that random real-valued states form an $\varepsilon$-approximate state $t$-design for $t < \sqrt{\varepsilon d}$.

\section{Exact Designs}
The moment operators change if we choose a different orbit other than that of $\ket{0}$. We will now show that, for $t=3$, we can choose an orbit such that $\Or$ forms an \textit{exact} state $t$-design. We will further show that this cannot be true for any $t>3$.

We start by illustrating the technique for $t=2$. Let the coefficient vector be $c_{\mathfrak{m}} = \frac{1}{P(d,t)}\id\{\mf{m} \in \mbb{S}_2\}$. Multiplying this by the gram matrix $G$ yields a set of polynomial constraints that the state $\ket{\psi}$ must satisfy. For $t=2$, $P(d,t)=d(d+1)$ and
\begin{align}
    G = \begin{pmatrix}
        d^2 & d & d\\
        d & d^2 & d\\
        d & d & d^2
    \end{pmatrix}\ ,\quad G\begin{pmatrix}
        \frac{1}{d(d+1)}\\
        \frac{1}{d(d+1)}\\
        0
    \end{pmatrix} = \begin{pmatrix}
        1\\
        1\\
        \frac{2}{d+1}
    \end{pmatrix}\ .
\end{align}
This yields the single non-trivial constraint $\vert \langle \psi^* \vert \psi \rangle \vert^2 = \frac{2}{d+1}$. As an example, this is satisfied by the state
\begin{align}
    \ket{\psi} & = \sqrt{\frac{1}{2}\left(1-\sqrt{\frac{2}{d+1}}\right)}\ket{0}+i\sqrt{\frac{1}{2}\left(1+\sqrt{\frac{2}{d+1}}\right)}\ket{1}\ .
\end{align}

\begin{proposition}\label{prop:3_design}
    If $\ket{\psi}\in\mbb{C}^d$ is such that $\vert \langle \psi^* \vert \psi \rangle \vert^2 = \frac{2}{d+1}$, then $\int_{\Or} (O\ket{\psi}\bra{\psi}O^\top)^{\otimes 3}dO$ is $\rho_{sym}$.
\end{proposition}
\begin{proof}
    Multiplying the coefficient fector $c_\sigma = \frac{1}{P(d,t)}\id\{\sigma \in \mbb{S}_t\}$ by the $15 \times 15$ Gram matrix yields this constraint.
\end{proof}
\begin{proposition}
    Let $t > 3$ and $d > 1$. Then, no state $\ket{\psi}\in\mbb{C}^d$ exist such that $\int_{\Or }(O\ket{\psi}\bra{\psi}O^\top)^{\otimes t}dO$ is exactly $\rho_{sym}$.
\end{proposition}
\begin{proof}
    We set $t=4$ as this would imply that no such state exists for any larger $t$ (simply by tracing out $t-4$ copies). Explicitly computing the $105 \times 105$ Gram matrix yields the constraints
    \begin{align}
        \vert \langle \psi^* \vert \psi \rangle \vert^2  = \frac{2}{d+1}\ , \quad
    \vert \langle \psi^* \vert \psi \rangle \vert^4  = \frac{8}{(d+1)(d+3)}\ .
    \end{align}
    But these can only simultaneously be true when $d=1$.
\end{proof}

\section{Discussion}
In this paper we showed that there is no sample efficient protocol to distinguish random real-valued and complex-valued pure quantum states. It would be interesting if this had some application in the study of quantum scrambling. We additionally showed that the orbits of states with complex coefficients may be a closer match to the full Haar measure on $\sh$. An interesting remaining question is how well these orbits approximate $\rho_{sym}$ for $t > 3$.

\begin{acknowledgments}
\textbf{Acknowledgments }
L.S. thanks Jacob Beckey, Eric Chitambar, Mart\'{i}n Larocca, Marco Cerezo, Felix Leditzky, Fernando Jeronimo, and Makrand Sinha for discussions. L.S. was supported by IBM through the IBM-Illinois Discovery Accelerator Institute.
\end{acknowledgments}

\newpage

\bibliographystyle{apsrev4-2}
\bibliography{refs.bib}

\end{document}